\documentclass[11pt,a4paper,reqno]{amsart}
\usepackage{amsthm,amsmath,amsfonts,amssymb,amsxtra,bookmark,dsfont,graphicx}

\usepackage[margin=35mm]{geometry}
\usepackage[utf8]{inputenc}

\theoremstyle{plain}
\newtheorem{theorem}{Theorem}
\newtheorem{conjecture}[theorem]{Conjecture}

\renewcommand\phi{\varphi}

\newcommand{\R}{\mathbb{R}}

\newcommand{\eps}{\varepsilon}

\def\d{\,{\rm d}}

\begin{document}

\title{The ionization problem in quantum mechanics}

\author[P. T. Nam]{Phan Th\`anh Nam}
\address{Department of Mathematics, LMU Munich, Theresienstrasse 39, 80333 Munich, Germany} 
\email{nam@math.lmu.de}

\maketitle

\begin{center}
{\em Dedicated to Elliott H. Lieb on the occasion of his 90th birthday} 
\end{center}

\begin{abstract}  While it is well-known experimentally that a neutral atom can bind at most
one or two extra electrons, deriving this fact rigorously from first principles of quantum mechanics remains a very challenging problem, often referred to as the ionization conjecture. We will review some of Elliott H. Lieb's fundamental contributions to this topic and discuss their impacts on several recent developments. 

%\medskip
%\noindent MSC2020: 81V45, 81V74, 81V73
%
%\medskip
%\noindent Keywords: Schr\"odinger equation, Thomas-Fermi theory, Hartree-Fock theory, liquid drop model 

%\bigskip
%
%\begin{center}{\em Dedicated to Elliott H. Lieb on the occasion of his 90th birthday.}\end{center}
\end{abstract}

\section{The ionization problem}\label{sec:1}

A great achievement of quantum mechanics, which goes back to its birthday,  is the satisfactory explanation of {\em the stability of atoms}. In this context, the fact that the electrons do not fall into the nucleus of an atom can be derived mathematically from a variant of {\em Heisenberg's uncertainty principle}  (e.g.  Hardy's or Sobolev's inequality). On the other hand, the  deeper question  ``How many electrons that a nucleus can bind?" has not been answered completely. From experiments it is widely believed that a neutral atom can bind at most
one or two extra electrons, but proving this fact mathematically from the many-body Schr\"odinger equation remains a very challenging problem.

In this article, we will  limit the discussion to the Born--Oppenheimer approximation of non-relativistic atoms. To be precise, we consider a system of $N$ quantum electrons moving around a  classical nucleus  and interacting via Coulomb forces. The statistical properties of electrons are encoded by a normalized wave function in $L^2(\R^{3N})$ which satisfies the anti-symmetry
\begin{equation} \label{eq:Pauli}
\Psi(x_1,...,x_i,...,x_j,...,x_N)= - \Psi(x_1,...,x_j, ..., x_i, ...,x_N), \quad \forall i\ne j,
\end{equation}
where $x_i\in \mathbb{R}^3$ stands for the position of the $i$-th electron (we ignore the spin of electrons for simplicity). Usually $|\Psi|^2$ is interpreted as the probability density of $N$ electrons. The condition \eqref{eq:Pauli} is called {\em Pauli's exclusion principle}; in particular it implies that $\Psi(x_1,...,x_i,...,x_j,...,x_N)=0$ if $x_i=x_j$ with $i\ne j$, namely two electrons cannot occupy a common position. As we will see,   Pauli's exclusion principle plays a crucial role in the ionization problem.

The Hamiltonian of the system is 
\begin{equation} \label{eq:HN}
H_{N} = \sum\limits_{i = 1}^N {\left( { - \Delta _{x_i}  - \frac{Z}
{{|x_i |}}} \right)}  + \sum\limits_{1 \le i < j \le N} {\frac{1}
{{|x_i  - x_j |}}} .
\end{equation}
Here we use atomic units; in particular the electronic charge is $-1$ and the nuclear charge is $Z\in \mathbb{N}$. By Hardy's inequality 
\begin{equation} \label{eq:Hardy}
-\Delta \ge \frac{1}{4|x|^2}\quad \text{on }L^2(\R^3)
\end{equation}
it is easy to see that $H_N$ is bounded from below with the core domain $C_c^\infty(\R^{3N})$, and hence it can be extended to be a self-adjoint operator by Friedrichs' method with the form domain $H^1(\R^{3N})$. We will always restrict $H_N$ to the anti-symmetric  subspace induced by \eqref{eq:Pauli}.

We are interested in the ground state problem 
\begin{equation} \label{eq:EN}
E_N = \inf \sigma(H_N)=\mathop {\inf }\limits_{||\Psi ||_{L^2}  = 1} \langle \Psi ,H_{N} \Psi \rangle. 
\end{equation}
Obviously, both $H_N$ and $E_N$ depend on $Z$, but we ignore this dependence in the notation. It is well-known that if a minimizer $\Psi$ of  \eqref{eq:EN} exists, then it is a solution to the Schr\"odinger equation
\begin{equation} \label{eq:EN-eq}
H_N \Psi = E_N \Psi.
\end{equation}

Although $E_N$ is finite, the existence of minimizers of \eqref{eq:EN} does not necessarily
hold. Heuristically, there should a transition when $N$ increases as follows: 
\begin{itemize}
\item If $N< Z+1$, then the outermost electron is attracted by the rest of the system which is of the effective charge $Z-(N-1)>0$. The existence is likely to hold in this case;

\item If $N>Z+1$, then the outermost electron prefers to ``escape to infinity" due to the Coulomb repulsion. The nonexistence is in favor. 
\end{itemize}
The first half of this prediction, namely the existence of all positive ions and neutral atoms, was proved by  Zhislin in 1960.  

\begin{theorem}[\cite{Z60}]  \label{thm:Zhislin} If $N<Z+1$, then $E_N$ has a minimizer. 
\end{theorem}

On the other hand, the second half of the above prediction, namely the nonexistence of highly negative ions, is much more difficult and often referred to as the ``ionization conjecture",  see \cite[Problem 9]{S00} and \cite[Chapter 12]{LS09}. To be precise, let us denote by $N_c=N_c(Z)$ the largest number of electrons such that $E_{N_c}$ has a minimizer. Then we have the following conjecture.

\begin{conjecture}[\cite{S00,LS09}] \label{con2} $N_c\le Z+C$ with a constant $C>0$ (possibly $C=1$).
\end{conjecture}

Due to the celebrated Hunziker--van Winter--Zhislin (HVZ) theorem (see e.g. \cite[Theorem 11.2]{T09}), the essential spectrum of $H_N$ is $[E_{N-1},\infty)$. Consequently, we always have $E_N\le E_{N-1}$; moreover, if $E_N<E_{N-1}$, then $H_N$ has a bound state with eigenvalue $E_N$.  In \cite{Z60}, Zhislin proved Theorem \ref{thm:Zhislin} by establishing the strict binding inequality $E_N<E_{N-1}$ for all $N<Z+1$ by an induction argument (if $E_{N-1}$ has an eigenfunction $\Psi_{N-1}$, then the energy of the $N$-body state $\Psi_{N-1}\wedge \varphi$ is lower than $E_{N-1}$ for some function $\varphi\in L^2(\R^3)$ describing one electron at infinity). In contrast, Conjecture \ref{con2} implies that there exists $N_c \le Z+C$ such that 
\begin{equation} \label{eq:EN-ENc}
E_N=E_{N_{c}},\quad \forall N\ge N_c.
\end{equation}
It is believed that $E_N$ is not only {\em strictly decreasing} when $N\le N_c$, but also {\em convex}. 

\begin{conjecture}[\cite{LS09}] \label{conCV} The function $E_N$ is convex in $N\in \mathbb{N}$.  
\end{conjecture}

A consequence of Conjecture \ref{conCV} is that if the nucleus can bind $N$ electrons ($E_N<E_{N-1}$), then it can also bind $N-1$ electrons ($E_{N-1}<E_{N-2}$).  This seemingly obvious fact is still an open mathematical question!

\bigskip

In the next sections, we will review some rigorous results on the ionization problem, in the full many-body Schr\"odinger theory as well as in some simplified models. We will summarize the main ideas and mention several open problems, thus extending a cordial invitation to young researchers to enter the subject. See also \cite{N21} for a shorter review on the same topic.

%The selection of results here may be not complete since it reflects my personal favourite, but hopefully it still serves as an invitation to young researchers to enter the subject. 

\bigskip

\section{Non-asymptotic bounds} \label{sec:2}

The fact that $N_c=N_c(Z)<\infty$ is already highly nontrivial. This was proved in 1982 independently by Ruskai \cite{R82,R82b} and Sigal  \cite{S82,S84}. While the conjectured bound $N_c(Z)\le Z+O(1)_{Z\to \infty}$ remains open, there are non-asymptotic bounds on $N_c(Z)$ and the most important one was proved by Lieb in 1984. 

\begin{theorem}[\cite{L84b,L84}] \label{thm:L84} We have $N_c(Z)<2Z+1$ for all  $Z>0$.

\end{theorem}

This result holds even if $Z$ is not an integer, and it can also be  extended to molecules. In particular, it settles the ionization conjecture for the hydrogen atom. In spite of its importance, the proof of Theorem \ref{thm:L84} is so short and elegant that, as recommended in \cite{L84b}, it can be given in any elementary course of quantum mechanics. 

\begin{proof}[Proof of Theorem \ref{thm:L84}] Assume that  Schr\"odinger's equation \eqref{eq:EN-eq} has a solution $\Psi$. Then multiplying the equation with $|x_N|$ we have 
\begin{align}  \label{eq:Lieb-proof-0}
0&=\langle |x_N| \Psi, (H_N -E_N) \Psi\rangle= \langle |x_N| \Psi, (H_{N-1} -E_N) \Psi\rangle  \nonumber\\ 
&\quad+  \Re \langle |x_N| \Psi, (-\Delta_{x_N}) \Psi\rangle - Z +\frac{1}{2}\sum_{i=1}^{N-1} \left\langle \Psi, \frac{|x_N|+|x_i|}{|x_i-x_N|} \Psi \right\rangle. 
\end{align}
Here we have used the symmetry of $|\Psi|^2$ to symmetrize the interaction term. The first two terms on the right hand side of \eqref{eq:Lieb-proof-0} can be dropped for a lower bound thanks to the obvious inequality $
H_{N-1} \ge E_{N-1}\ge E_N
$ for the first $(N-1)$ electrons and the operator inequality
\begin{equation} \label{eq:Lieb-ineq}
(-\Delta)|x|+|x|(-\Delta)\ge 0\text{ on } L^2(\mathbb{R}^3)
\end{equation}
for the $N$-th electron.  Combining with the triangle inequality $|x|+|y| \ge |x-y|$ we conclude from \eqref{eq:Lieb-proof-0} that 
$ 0 > - Z + \frac{N-1}{2},  
$ namely $N<2Z+1$. Here we get the strict inequality because the triangle inequality is strict almost everywhere.  \end{proof}

The one-body inequality \eqref{eq:Lieb-ineq} was proved directly in \cite{L84b,L84} by integration by parts. As explained in \cite{N12}, this bound can be also deduced from Hardy's inequality \eqref{eq:Hardy} by applying the IMS formula 
\begin{equation} \label{eq:IMS}
\frac{\varphi(x)^2 (-\Delta)+ (-\Delta)\varphi(x)^2}{2}= \varphi(x) (-\Delta) \varphi (x) - |\nabla \varphi|^2
\end{equation}
to the case $\varphi(x)= |x|^{1/2}$. In 2013, Chen and Siedentop \cite{CS13} proved an interesting generalization of \eqref{eq:IMS}: if $\min(a,b)\in [0,2]$ and $a+b\le d$, then  
$$
|\nabla_x|^a |x|^b + |x|^b |\nabla_x|^a\ge 0 \text { on }L^2(\R^d). 
$$

The method of ``multiplying the equation by $|x|$" is called the {\em Benguria--Lieb argument}. It was used by Benguria on simplified models \cite{B79} and extended by Lieb to the full many-body context. Nowadays, this is a standard argument in the analysis of Coulomb systems. Let us discuss below two further results obtained by variants of this argument. 

In 2012, we derived a new bound which improves  Theorem \ref{thm:L84} for $Z\ge 6$. 

\begin{theorem}[\cite{N12}] \label{thm:N12} We have $N_c(Z)<1.22 Z+3 Z^{1/3}$ for all $Z\ge 1$. 
\end{theorem} 

\begin{proof}[Ideas of the proof] Multiplying Schr\"odinger's equation \eqref{eq:EN-eq}  with $|x_N|^2$ (instead of $|x_N|$) we have
\begin{align*}  
0&=\langle |x_N|^2 \Psi, (H_N -E_N) \Psi\rangle = \langle |x_N|^2 \Psi, (H_{N-1} -E_N) \Psi\rangle \\
&\quad +  \Re \langle |x_N|^2 \Psi, (-\Delta_{x_N}) \Psi\rangle - Z  \left\langle \Psi, |x_N| \Psi \right\rangle   + \frac{1}{2}\sum_{i=1}^{N-1} \left\langle \Psi, \frac{|x_N|^2+|x_i|^2}{|x_i-x_N|} \Psi \right\rangle. 
\end{align*}
The first term on the right hand side can be dropped for a lower bound as before. For the second term, by the IMS formula \eqref{eq:IMS} and Hardy's inequality \eqref{eq:Hardy}   we have
\begin{equation} \label{eq:IMS-x2}
\frac{|x|^2(-\Delta)+ (-\Delta)|x|^2}{2}= |x|(-\Delta) |x| - 1  \ge \frac{1}{4}-1= -\frac{3}{8}. 
\end{equation}
The error in \eqref{eq:IMS-x2} is small in comparison with the third term thanks to the bound 
$$
Z \left\langle \Psi, |x_N| \Psi \right\rangle > 0.553 N^{2/3} 
$$
which is a consequence of the Lieb-Thirring inequality \cite{LT75}. On the other hand, the last term on the right hand side, up to a symmetrization, can be estimated by 
$$
\beta \ge \mathop {\inf }\limits_{\{x_i\}_{i=1}^N\subset \mathbb{R}^3 } \frac{{\sum\limits_{1 \le i < j \le N} {\frac{{|x_i |^2  + |x_j |^2 }}
{{|x_i  - x_j |}}} }}
{{N(N-1)\sum\limits_{i = 1}^N {|x_i |} }} \ge  \beta -1.55 N^{-2/3}
$$
where $\beta$ is determined by a variational problem of infinitely many classical particles
$$
\beta:= \inf_{\substack{\rm \rho ~probability\\ \rm ~measure~in~\mathbb{R}^3}} \left\{ {\frac{{\iint\limits_{\mathbb{R}^3  \times \mathbb{R}^3 } {\frac{{x^2+y^2}}
{{2|x - y|}} {{\rm d}\rho} (x){{\rm d}\rho} (y)}}}
{{\int\limits_{\mathbb{R}^3 } {|x|  {{\rm d}\rho} (x)} } }} \right\}. 
$$
The key improvement comes from the fact that $\beta \ge 0.82$ (instead of $1/2$ by triangle inequality). This eventually implies that $N<1.22 \,Z + 3 Z^{1/3}$ (here $\beta^{-1} \approx 1.22$). 
\end{proof}

In 2013, Lenzmann and Lewin proved a stronger version of the nonexistence, where the absence of not only ground states but also all eigenfunctions  is concerned. 

\begin{theorem}[\cite{LL13}] \label{thm:LL13} If $N\ge 4Z+1$, then $H_N$ has no eigenvalue. 
\end{theorem}

\begin{proof}[Ideas of the proof] If $\Psi$ is an eigenfunction of $H_N$, then for every one-body self-adjoint operator $A$ on $L^2(\R^3)$ we have 
\begin{equation} \label{eq:LL-proof-0-1}
0 = \langle \Psi,  {\bf i} [H_N, A_{x_N}] \Psi\rangle =  \left\langle \Psi, {\bf i} \left[-\Delta_{x_N} -\frac{Z}{|x_N|} + \sum_{j=1}^N \frac{1}{|x_j-x_N|}, A_{x_N} \right] \Psi \right\rangle 
\end{equation}
with ${\bf i}^2=-1$. In particular, choosing 
$$
A= {\bf i}[\Delta, f(x)]= ({\bf i} \nabla_{x})\cdot \nabla f(x) + \nabla f(x) \cdot ({\bf i} \nabla_{x}) 
$$
we find that
$$
0 =  \left\langle \Psi, [\Delta_{x_N},[\Delta_{x_N},f(x_N)]] \Psi \right\rangle + \left\langle \Psi, \nabla f (x_N)   \cdot \left( \frac{Z x_N}{|x_N|^3} -  \sum_{j=1}^{N-1}  \frac{x_N-x_j}{|x_N-x_j|^3}\right) \Psi \right\rangle.  
$$
With the special choice $f(x)=\frac{1}{3}|x|^3$ we have $\nabla f (x) = |x| x$, and hence 
$$
0 =  \frac{1}{3}\left\langle \Psi, [\Delta_{x_N},[\Delta_{x_N},|x_N|^3]] \Psi \right\rangle + Z -  \frac{1}{2}\left\langle \Psi,  \sum_{j=1}^{N-1} \frac{(|x_j| x_j - |x_N| x_N)\cdot (x_j-x_N) }{|x_j-x_N|^3} \Psi \right\rangle
$$
where we used the symmetrization for the interaction term. The latter identity is very similar to \eqref{eq:Lieb-proof-0}. Then the conclusion follows from two key ingredients: first 
$$
[\Delta,[\Delta,|x|^3]] \le 0 \text{ on }L^2(\R^3)
$$
which should be compared with  \eqref{eq:Lieb-ineq}, and second 
$$
\inf_{x\ne y\in \R^3}\frac{(|x|x- |y|y)\cdot (x-y)}{|x-y|^3} = \frac{1}{2}
$$
which should be compared with the triangle inequality. Thus $0>Z+\frac{N-1}{4}$, namely $N<4Z+1$. 

Strictly speaking, some decay condition on $\Psi$ is needed  to ensure that all relevant quantities are finite. However, this technical condition can be relaxed by choosing $f_R(x)=R^3 g(|x|/R)$ with $g(r)=r-\arctan r$ and then sending $R\to \infty$. 
\end{proof}

In the above proof, it is helpful to  interpret   \eqref{eq:LL-proof-0-1} as a stationary condition for the Schr\"odinger dynamics $\Psi(t)=e^{-{\bf i} tH_N} \Psi$: 
\begin{align*}
 0 = \frac{\d^2}{\d t^2} \langle \Psi(t), f(x_N) \Psi(t)\rangle =  \langle \Psi(t), -[H_N,[H_N, f(x_N)]] \Psi(t)\rangle,
\end{align*}
which explains the choice of $A= {\bf i} [\Delta, f(x)]$. This time-dependent technique goes back to the famous {\em Morawetz--Lin--Strauss estimate} for nonlinear Schr\"odinger equations (NLS). In the standard NLS, the choices $f(x)=|x|, |x|^2, |x|^4$ were used in \cite{M68,G77,T08}, respectively. The argument in \cite{LL13} shows that the new choice $f(x)=|x|^3$ corresponds to the time-dependent version of Lieb's proof in \cite{L84b,L84}. 

The above approach motivates a stronger version of Conjecture \ref{con2}. 

\begin{conjecture} \label{eq:Con-stronger}  There exists a universal constant $C>0$ such that if $N>Z+C$, then $H_N$ has no eigenvalue. 
\end{conjecture}

\section{Stability of bosonic atoms}

We have mentioned in the introduction that the ionization problem is strongly associated to Pauli's exclusion principle. However, we have not seen this subtle fact so far.   The heuristic idea supporting for Conjecture \ref{con2} relies only on an electrostatic argument which is purely classical. Hence, in principle it applies to not only  {\em anti-symmetric} wave functions $\Psi\in L^2(\R^{3N})$ but also  {\em totally symmetric} ones, namely
\begin{equation} \label{eq:Pauli-off}
\Psi(x_1,...,x_i,...,x_j,...,x_N)= \Psi(x_1,...,x_j, ..., x_i, ...,x_N), \quad \forall i\ne j.
\end{equation}
The latter case corresponds to the so-called ``bosonic atoms" where electrons are treated as if they were bosonic particles. 

Note that all of the HVZ theorem, Zhislin's theorem (Theorem \ref{thm:Zhislin}) and Lieb's theorem (Theorem \ref{thm:L84})  work equally well for bosonic atoms. In contrast,  in 1983, Benguria and Lieb proved a striking result  that Conjecture \ref{con2} {\em fails} in the bosonic case, thus firmly validating the importance of Pauli's exclusion principle   \eqref{eq:Pauli} in the ionization problem. 

\begin{theorem}[\cite{BL83}] \label{thm:BL83} If \eqref{eq:Pauli} is replaced by the bosonic symmetry \eqref{eq:Pauli-off}, then 
$$\liminf_{Z\to \infty} \frac{N_c(Z)}{Z} \ge t_c >1. 
$$
\end{theorem}

As we will see below, the constant $t_c$ is taken from Hartree's theory, which is known numerically $t_c \approx 1.21$ \cite{B84}. In 1990, Solovej \cite{S90b} proved the optimality of Theorem \ref{thm:BL83}  by providing a matching asymptotic upper bound, namely $N_c/Z\to t_c$. 
%Asymptotics for bosonic atoms, Lett. Math. Phys., 20, 165-172, 1990

In principle, it is also natural to consider the variational problem \eqref{eq:EN} without any symmetry condition, but this results in the same problem with the bosonic symmetry \eqref{eq:Pauli-off}, see e.g. \cite[Corollary 3.1]{LS09}.

%Thus Conjecture \ref{con2} fails if Pauli's exclusion principle is turned off. 
%
%the stability of atoms is improved greatly in this bosonic setting. 

\begin{proof}[Proof of Theorem \ref{thm:BL83}] 

The main principle of the proof is that the many-body energy $E_N$ in \eqref{eq:EN} with the symmetry condition \eqref{eq:Pauli-off}  can be approximated by {Hartree's theory} where one restricts to the uncorrelated ansatz
\begin{align}\label{eq:Hartree}
\Psi(x_1,...,x_N)= (u^{\otimes N})(x_1,...,x_N) = u(x_1)... u(x_N) 
\end{align}
for a normalized function $u\in L^2(\R^3)$. This leads to Hartree's energy   
\begin{align} \label{eq:eH}
E^{\rm H}_N=&\inf_{\|u\|_{L^2}=1} \langle u^{\otimes N}, H_N u^{\otimes N}\rangle \nonumber\\
&=  \inf_{\|u\|_{L^2}=1}  N \int_{\R^3} \left(|\nabla u(x)|^2 - \frac{Z |u(x)|^2}{|x|} + \frac{N-1}{2} |u(x)|^2 \left (|u|^2 * \frac{1}{|x|} \right) \right) \d x.  
\end{align}
Note that by rescaling $u(x)= t^{-1/2} Z^{-3}v(x/Z)$ with $t=(N-1)/Z$ we can write 
\begin{align} \label{eq:eH}
E^{\rm H}_N = \frac{N Z^3}{N-1} e(t)
\end{align}
where
\begin{align} \label{eq:eHt}
e(t)= \inf_{\|v\|^2_{L^2}=t}  \int_{\R^3} \left(|\nabla v(x)|^2 - \frac{|v(x)|^2}{|x|} + \frac{1}{2} |v(x)|^2 \left (|v|^2 * \frac{1}{|x|} \right) \right) \d x .
\end{align}
The existence/nonexistence within Hartree's theory is fairly easy to handle since the functional on the right hand side of \eqref{eq:eHt}  is convex in $|v|^2$. In particular, it is well-known that $e(t)$ has a minimizer if and only if $t\le t_c$ for a constant $t_c>1$; moreover $e(t)$ is negative and strictly decreasing when $t\le t_c$ while $e(t)=e(t_c)$ for all $t\ge t_c$ (see \cite[Lemma 13]{BBL81} and \cite[Theorem 7.16]{L81}).

The main difficulty here is to justify Hartree's approximation. The upper bound $E_N\le E_N^{\rm H}$ follows directly by the variational principle, but obtaining a good  lower bound is not obvious. This follows from the {Hoffmann--Ostenhof inequality} \cite{HO77} 
$$
K=\left\langle \Psi, \sum_{i=1}^N (-\Delta_{x_i})\Psi \right\rangle \ge \int_{\R^3} |\nabla \sqrt{\rho_{\Psi}}|^2 
$$
and the {Lieb--Oxford inequality} \cite{LO81} 
\begin{align} \label{eq:LO}
\left\langle \Psi, \sum_{1\le i<j\le N} \frac{1}{|x_i-x_j|} \Psi \right\rangle \ge \frac{1}{2}\int_{\R^3} \frac{\rho_\Psi(x)\rho_\Psi(y)}{|x-y|} \d x \d y - 1.68 \int_{\R^3}\rho_\Psi^{4/3}
\end{align}
where $\rho_\Psi$ is the one-body density of the $N$-body wave function $\Psi$,
\begin{align} \label{eq:rhoN}
\rho_\Psi(x) = N \int_{\R^{3(N-1)}} |\Psi(x,x_2,...,x_N)| \d x_2 ... \d x_N.
\end{align}
Note that $\rho\ge 0$ and $\int_{\R^3} \rho=N$ since $\Psi$ is normalized. The error term in \eqref{eq:LO} can be controlled further by H\"older's and Sobolev's inequalities 
$$
\int_{\R^3} \rho_\Psi^{4/3} \le \left( \int_{\R^3} \rho_\Psi \right)^{5/6} \left( \int_{\R^3} \rho_\Psi^3 \right)^{1/6} \le C N^{5/6} K^{1/2}.  
$$
By Hardy's inequality \eqref{eq:Hardy} it is easy to show that $E_N \ge - CN Z^{2}$. Consequently, for a ground state $\Psi$ of $H_N$, the kinetic energy is controlled as $K\le CNZ^{2}$. Moreover, thanks to Lieb's theorem (Theorem \ref{thm:L84}), the existence of minimizer implies that $N\le 2Z+1$. All this gives 
\begin{align} \label{eq:EN-boson-lowerbound}
E_N &\ge   \int_{\R^3} \left( |\nabla \sqrt{\rho_{\Psi}}|^2  - \frac{Z  \rho_\Psi(x)}{|x|} + \frac{1}{2} \rho_\Psi(x) \left ( \rho_\Psi * \frac{1}{|x|} \right) \right) \d x - CN^{5/6} (NZ^2)^{1/2}  \noindent \nonumber\\
&\ge Z^3 \Big( e(N/Z)  - CZ^{-2/3}\Big). 
\end{align}

In summary, if $N_c \le t_c Z+1$, then from \eqref{eq:EN-ENc},  \eqref{eq:eH},  and \eqref{eq:EN-boson-lowerbound} we find that  
$$Z^3 \Big( e(N_c/Z)  - CZ^{-2/3}\Big) \le E_{N_c}=E_{t_c Z+1} \le Z^3 e(t_c),$$
and hence
$$
\limsup_{Z\to \infty} e(N_c/Z) \le e(t_c). 
$$
Since $e(t)$ is strictly decreasing when $t\le t_c$, we obtain the desired result 
$$\liminf_{Z\to \infty} N_c/Z = t_c.$$  
\end{proof}

In fact, the influence of Benguria and Lieb's argument in \cite{BL83} goes far beyond the context of the ionization problem. It has inspired several works dealing with the justification of  Hartree's theory from  many-body bosonic systems, which is  particularly relevant to the description of the Bose--Einstein condensation for interacting Bose gases. We refer to \cite{LNR14}  for further discussions.

Now let us focus on Hartree's theory. Note that the bound $t_c<2$ follows easily from the Benguria--Lieb argument as in  Theorem \ref{thm:L84}. However, it is not easy to improve. Very recently,  Benguria and Tubino \cite{BT22} successfully proved that $t_c<1.5211$ (thus approaching closer to the numerical value $t_c \approx 1.21$ \cite{B84}). Their proof strategy is similar to that of Theorem \ref{thm:N12}, but they were able to replace the use of the Lieb--Thirring inequality (which works only for fermions) by a clever application of the virial theorem to control the kinetic energy.

Finally let us mention the following analogue of Conjecture \ref{eq:Con-stronger} for Hartree's equation, which is essentially taken from \cite{LL13}. 

\begin{conjecture} \label{con:Hartree-hard} If Hartree's equation 
$$
(-\Delta - |x|^{-1} + |u|^2* |x|^{-1}) u= \lambda u 
$$
has a solution $u\in H^1(\R^3)$ with a constant $\lambda\in \mathbb{R}$, then 
$\int_{\R^3}|u|^2 \le t_c.$
\end{conjecture}

%
%Lenzmann and Lewin proved a stronger version of the nonexistence, where the absence of not only ground states but also all eigenfunctions  is concerned. 
%
%\begin{theorem}[\cite{LL13}] 

In \cite{LL13}, Lenzmann and Lewin proved a weaker bound with $4t_c$ instead of  $t_c$, using  the same proof strategy of Theorem \ref{thm:LL13}. In fact, they considered the following dynamical version of Conjecture \ref{con:Hartree-hard} and proved the bound $4t_c$ in this stronger sense. 

\begin{conjecture} \label{con:Hartree-hard-time} Consider the time-dependent Hartree's equation 
$$
{\bf i} \partial_t u= (-\Delta - |x|^{-1} + |u|^2* |x|^{-1}) u. 
$$
Then for every initial state $u_0\in H^1(\R^3)$,  we have
$$
\limsup_{T\to \infty}\frac{1}{T} \int_0^T \int_{|x|\le R} |u(x,t)|^2 \d x \d t \le t_c, \quad \forall R>0. 
$$
\end{conjecture}

The idea here is that even if we start from an initial state with arbitrarily large mass, when the time becomes large, there is only at most $t_c$ mass staying in every bounded domain.  Conjecture \ref{con:Hartree-hard-time} is closely related to the scattering theory of dispersive PDEs with long-range interaction potentials, which is a very interesting topic in its own right.

\section{Asymptotic neutrality}

Now we come back to the ionization problem with Pauli's exclusion principle \eqref{eq:Pauli}. A fundamental step towards the ionization conjecture is the following result of Lieb, Sigal, Simon and Thirring in 1984. 

\begin{theorem}[\cite{LSST84,LSST88}] \label{thm:asympneu}We have
$$
\lim_{Z\to \infty} \frac{N_c(Z)}{Z} = 1. 
$$
\end{theorem}

\begin{proof}[Ideas of the proof] The general strategy goes back to the geometric localization method used in Sigal's proof of $N_c(Z)<\infty$ in \cite{S82}. The idea is that by introducing a suitable partition of unity of $(\R^{3})^{N}$, the quantum problem can be reduced to a classical problem. In \cite{S84}, Sigal derived the asymptotic bound
\begin{equation} \label{eq:Sigal-0} \limsup_{Z\to \infty}\frac{N_c(Z)}{Z}\le 2
\end{equation}
using its classical analogue
\begin{equation} \label{eq:Sigal}
\mathop {\max }\limits_{1\le j\le N} \left\{ {\sum\limits_{1 \le i \le N,i \ne j} {\frac{1}
{{|x_i  - x_j |}}}  - \frac{Z}
{{|x_j |}}} \right\} \ge  0
\end{equation}
which follows easily from the triangle inequality. The key ingredient in \cite{LSST84} is the following improvement of \eqref{eq:Sigal}: for every $\eps>0$, $N\ge N_\eps$, and $\{x_i\}_{i=1}^N\subset \R^3$ we have
\begin{equation} \label{eq:Sigal-improved}
\mathop {\max }\limits_{1\le j\le N} \left\{ {\sum\limits_{i \ne j} {\frac{1}
{{|x_i  - x_j |}}}  - \frac{(1-\eps)N}
{{|x_j |}}} \right\} \ge 0.
\end{equation}
By a contradiction argument, \eqref{eq:Sigal-improved} can be deduced from its continuum analogue which is a nice result in potential  theory: for any probability measure $\mu\ne \delta_0(x)$ in $\R^3$ and for any $\eps>0$, there exists a point $x \in {\rm supp} (\mu)\backslash\{0\}$ such that 
\begin{equation} \label{eq:Sigal-improved-2}
f(x)=\int_{\R^3} \frac{1}{|x-y|} \d \mu(y) -  \frac{1-\eps}{|x|} \ge 0. 
\end{equation}
In the simple case, if ${\rm supp} (\mu)$ is bounded and does not contain $0$, then $f$ is harmonic outside ${\rm supp} (\mu)$ and vanishing at infinity. Hence, if we assume further that $f$ is continuous, then $f$ must be nonnegative somewhere in ${\rm supp} (\mu)$ by the maximum principle. For the general case see \cite{LSST84,LSST88} for details.

For every $\eps>0$, $N\ge N_\eps$, and $R>0$, the inequality \eqref{eq:Sigal-improved} and its refinements allow to construct a partition of unity $\{J_a\}_{a=0}^N$ of $C^\infty$ functions in $(\R^{3})^N$ so that the following hold. 
\begin{itemize}
\item[(i)]  $J_a \ge 0$ for all $a$ and $\sum_{a=0}^N J_a^2(X)=1$ for all $X= \{x_b\}_{b=1}^{N} \subset (\R^{3})^N$. 

\item[(ii)] Denoting $|X|_\infty =\max_{1\le b\le N} |x_b|$ we have
$$L=\sum_{a=0}^N |\nabla_{\R^{3N}}J_a(X)|^2 \le C_\eps \frac{N^{1/2}\log(N)^2}{R |X|_\infty}.$$ 

\item[(iii)] $J_0$ is totally symmetric in $\{x_b\}_{b=1}^{N}$ and 
$${\rm supp}\, J_0 \subset \left\{ X= \{x_b\}_{b=1}^{N} \,|\,  |X|_\infty  \le R \right\}.$$ 

\item[(iv)] If $a\ne 0$, then $J_a$ is symmetric in $\{x_b\}_{b\ne a}$ and 
$${\rm supp}\, J_a \subset \left\{ X =\{x_b\}_{b=1}^{N} \,|\,   \sum_{b\ne a} \frac{1}{|x_b-x_a|} \ge \frac{(1-\eps)N}{|x_a|} \right\}.$$
\end{itemize}

Now let us conclude the proof of Theorem \ref{thm:asympneu}. Assume that 
$$N\in [(1-2\eps )^{-1}Z,2Z+1]$$ for some $\eps>0$ small (independent of $Z$). We choose 
$$
Z^{-1/2} \log(Z)^2 \ll R \ll Z^{-1/3}.
$$
By the IMS localization formula (c.f. \eqref{eq:IMS}) we can decompose $H_N$ in \eqref{eq:HN} as
\begin{equation}\label{eq:HN-decompose}
H_N = \sum_{a=0}^N \frac{1}{2} (J_a^2 H_N + H_N J_a^2) = \sum_{a=0}^N J_a H_N J_a  - L =  \sum_{a=0}^N J_a (H_N -L) J_a. 
\end{equation}
When $a=N$, by (ii) we have 
$$
L \le C_\eps \frac{Z^{1/2}\log(Z)^2}{R |X|_\infty} \ll \frac{\eps N}{|x_N|}. 
$$
Moreover, by the property of the support of $J_N$ in (iv) 
$$
J_N \sum_{i=1}^{N-1} \frac{1}{|x_i-x_N|} J_N \ge J_a   \frac{(1-\eps)N}{|x_N|} J_a.
$$
By decomposing 
$$
H_N -L = H_{N_1} -\Delta_{x_N} -\frac{Z}{|x_N|} + \sum_{i=1}^{N-1} \frac{1}{|x_i-x_N|} - L,
$$
then using $H_{N_1}\ge E_{N-1}$ and $-\Delta_{x_N}\ge 0$ we get 
$$
J_N (H_N - L) J_N \ge J_N \left( E_{N-1} +  \frac{(1-2\eps) N - Z}{|x_N|} \right) J_N \ge J_N^2 E_{N-1}. 
$$
Similarly,  $J_a (H_N - L) J_a\ge J_a^2 E_{N-1}$ for all $a\ne 0$. For $a=0$, we use (ii) again for $L$ and use the triangle inequality to control the interaction energy. This gives  
$$
J_0^2 \left( \sum_{1\le i<j\le N} \frac{1}{|x_i-x_j|} - L \right)  \ge J_0^2  \left(  \frac{N(N-1)}{2|X|_\infty} - C_\eps \frac{Z^{1/2}\log(Z)^2}{R |X|_\infty} \right)  \ge J_0 \frac{Z^2}{4R}  
$$
where  in the last inequality we also used the  property of the support of $J_0$ in (iii). On the other hand, by summing the first $N$ eigenvalues of the hydrogen atom, we have
$$
\sum_{i=1}^N \Big( -\Delta_{x_i} - \frac{Z}{|x_i|} \Big) \ge - C Z^{7/3}. 
$$
Here Pauli's exclusion principle is crucial. Thus 
$$
J_0 (H_N-L)J_0 \ge J_0^2 \left( \frac{Z^2}{4R}  - C Z^{7/3}  \right) \ge 0 \ge J_0^2 E_{N-1}.  
$$
In summary, from \eqref{eq:HN-decompose} we deduce that 
\begin{equation}\label{eq:HN-decompose}
H_N =  \sum_{a=0}^N J_a (H_N -L) J_a \ge  \sum_{a=0}^N J_a^2 E_{N-1} = E_{N-1} \ge E_N. 
\end{equation}
Therefore, if $H_N \Psi=E_N \Psi$ for a ground state $\Psi$, then we must have 
$$E_N=\langle \Psi, H_N \Psi\rangle=E_{N-1}.$$ A careful reconsideration of the proof of the above estimates shows that this equality cannot hold for any eigenfunction. Thus $N\in [(1-2\eps )^{-1}Z,2Z+1]$ fails, which implies the desired result. 
 \end{proof}
 
The convergence result in Theorem \ref{thm:asympneu} is not quantitative. In 1990, Fefferman and Seco \cite{FS90} and Seco, Sigal and Solovej \cite{SSS90} offered different proofs for the quantitative bound 
 $$
 N_c(Z) \le Z + O(Z^{5/7})
 $$
 which remains the best known asymptotic result. The proofs in \cite{FS90,SSS90} use certain information of the minimizing wave function obtained via a careful comparison with the Thomas--Fermi (TF) theory which will be revisited below.

% Both approaches in \cite{FS90,SSS90} are based on careful comparisons with the {\em Thomas--Fermi theory} (that we will revisit below). 
% 
% 
%This bound was obtained by comparing it with the
%%Thomas–Fermi theory (that we will revisit below)
%
%
%%Solovej [20],
% 
%Later, Lieb, Sigal, Simon and Thirring \cite{LSST84} found the following improvement: for every $\{x_i\}_{i=1}^N \subset \mathbb{R}^3$ with $N$ large,  one has
%\begin{equation}\label{eq:LSST-classical-bound}
%\mathop {\max }\limits_{1\le j\le N} \left\{ {\sum\limits_{1 \le i \le N,i \ne j} {\frac{1}
%{{|x_i  - x_j |}}}  - \frac{{N + o(N)}}
%{{|x_j |}}} \right\} \ge 0.
%\end{equation}
%Consequently, they obtained the asymptotic neutrality 
%$$
%\lim_{Z\to \infty}\frac{N_c(Z)}{Z} =1. 
%$$
%It is unclear whether one can improve the quantity $N+o(N)$ in \eqref{eq:LSST-classical-bound} to $N+O(N^\alpha)$ with some constant $0\le \alpha<1$. 
%
%%In 1990, Fefferman and Seco [4], and Seco, Sigal and
%%Solovej [20], proved
%%Theorem 3. When Z , we have Nc(Z)  Z + O(Z5/7).
%
%
%%This bound was obtained by comparing it with the
%%Thomas–Fermi theory (that we will revisit below) and taking
%%into account quantitative estimates for Scott’s correction
%%(studied by Huges, and by Siedentop andWeikard). There has
%%been no further improvement in the past three decades!

\section{The TF theory}

Since the many-body Schr\"odinger equation is too complicated, for practical computations one often replaces the wave function $\Psi$ by its one-body density $\rho_\Psi$ defined in \eqref{eq:rhoN}, resulting in density functional theories. A famous example is the TF theory, in which the ground state energy $E_N$ in \eqref{eq:EN} is replaced by its semiclassical approximation  
\begin{align}\label{eq:TF-theory}
E^{\rm TF}(N)= \inf_{\rho\ge 0, \int \rho=N}   \int_{\mathbb{R}^3} \Big( C^{\rm TF} \rho^{5/3}(x) - \frac{Z}{|x|} \rho(x) + \frac{1}{2} \rho(x) \Big(\rho*\frac{1}{|x|} \Big) \Big) \d x 
\end{align}
with a constant $C^{\rm TF}>0$. Here $N,Z>0$ are not necessarily integers.

The mathematical properties of the TF theory were studied in full detail by Lieb and Simon in 1973 \cite{LS73,LS77}.   In particular, concerning the ionization problem we have the following theorem.

\begin{theorem}[\cite{LS77}] For all $Z>0$, $E^{\rm TF}(N)$ has a minimizer if and only if $N\le Z$.  
\end{theorem}

\begin{proof}[Ideas of the proof]  A very useful argument introduced in  \cite{LS77} is the {\em relaxation method} which relates the variational problem \eqref{eq:TF-theory} to the ``unconstrained  problem"
\begin{align}\label{eq:TF-theory<=}
E^{\rm TF}_{\le}( N)= \inf_{\rho\ge 0, \int \rho \le N}   \int_{\mathbb{R}^3} \Big( C^{\rm TF} \rho^{5/3}(x) - \frac{Z}{|x|} \rho(x) + \frac{1}{2} \rho(x) \Big(\rho*\frac{1}{|x|} \Big) \Big) \d x. 
\end{align}
The advantage of \eqref{eq:TF-theory<=} is that the set of states is {\em convex}, and hence the existence of minimizers of \eqref{eq:TF-theory<=} follows easily by  the direct method in the calculus of variations (in particular, the set of states is stable under the weak convergence in $L^{5/3}$). Moreover, as explained in  \cite{LS77}, by standard rearrangement inequalities it is easy to see that the functional on the right hand side of \eqref{eq:TF-theory} is strictly convex. This implies that the unconstrained minimizer $\rho$ of \eqref{eq:TF-theory<=} is unique with $\int_{\R^3}\rho \le N$. Thus the existence  of a minimizer of the original problem \eqref{eq:TF-theory} is equivalent to $\int_{\R^3}\rho =N$. 

The existence part is rather standard: if $\int_{\R^3}\rho < N\le Z$, then we can put some positive mass at infinity to lower the energy and obtain a contradiction. 

Let us focus on the nonexistence part which is more challenging. As proved in \cite{LS77}, the unconstrained minimizer solves the TF equation 
$$
\frac{5}{3}C^{\rm TF}\rho(x)^{2/3} = [\Phi(x)]_+, \quad \Phi(x)= Z|x|^{-1}-\rho*|x|^{-1} -\mu 
$$
for some chemical potential $\mu \ge 0$. We prove that $\int_{\R^3}\rho \le Z$ for all $N>0$, which implies the nonexistence when $N>Z$. Assume by contradiction that $\int_{\R^3}\rho > Z$. Then the TF potential satisfies 
$$
|x| \Phi(x) \le Z -|x| \Big( \rho* \frac{1}{|x|}\Big) \to Z -\int_{\R^3} \rho<0
$$
as $|x|\to \infty$ (we used $\mu\ge 0$). Therefore, $A=\{x: \Phi(x)<0\}$ is non empty. Moreover, since $\Phi$ is continuous on $\R^3\backslash\{0\}$ and $\Phi(x)\to \infty$ as $|x|\to 0$, we find that $A$ is open and $0\ne A$. To conclude,  using $\Delta |x|^{-1} = 4\pi \delta_0(x)$ and the TF equation, we find that 
$$\Delta \Phi(x)= 4\pi \rho(x) = 0 \text { in } A.$$
Thus $\Phi$ is harmonic in $A$, $\Phi<0$ in $A$ and $\Phi=0$ on the boundary of $A$. All this leads to a contradiction to the maximum principle. Thus $\int_{\R^3} \rho\le Z$. 
\end{proof}

In fact, by a variant of the Benguria--Lieb argument, the nonexistence part can be also proved differently, in which we only use harmonic analysis via {\em Newton's theorem}. 

\begin{proof}[Another proof of $N\le Z$  \cite{N13}] Let $\rho$ be the unconstrained minimizer. Integrating the TF equation with $|x|^k\rho(x)$, $k\ge 1$, we get 
$$
0 \le \frac 5 3 C^{\rm TF} \int_{\R^3} \rho(x)^{5/3} |x|^k \d x =  \int_{\R^3} \Big(Z|x|^{-1}-\rho*|x|^{-1} -\mu \Big) \rho(x) |x|^k \d x. 
$$
The contribution associated with $\mu \ge 0$ can be ignored for an upper bound. Since the TF functional is rotation invariant and $\rho$ is unique, it must be radial. Hence, by Newton's theorem (see e.g. \cite[Theorem 5.2]{LS09}) we have
$$
\rho*|x|^{-1}= \int_{\mathbb{R}^3}  \frac{\rho(y)}{\max(|x|, |y|)} \d y. 
$$
Consequently,
\begin{align*}
Z \int_{\R^3} |x|^{k-1} \rho(x) &\ge  \int_{\R^3} |x|^k \rho(x) (\rho*|x|^{-1}) \d x = \frac{1}{2} \iint  \frac{ (|x|^k+|y|^k) \rho(x) \rho(y)}{\max(|x|, |y|)} \d x \d y. 
\end{align*}
Thanks to the AM-GM inequality,  
\[
\frac{{|x{|^k} + |y{|^k}}}{{\max( |x|,|y|) }} \ge \left( {1 - \frac{1}{k}} \right)\left( {|x{|^{k - 1}} + |y{|^{k - 1}}} \right), 
\]
and hence  
$$
Z \int_{\R^3} |x|^{k-1} \rho(x) \d x  \ge \left(1-\frac{1}{k}\right) \Big(  \int  |x|^{k-1} \rho(x) \d x \Big) \Big( {\int  \rho(y) {\rm d}y }\Big).$$
Thus $Z\ge (1-k^{-1})N$. The conclusion follows by taking $k\to \infty$. 

Strictly speaking, in the above proof we need $\int_{\R^3} |x|^{k-1} \rho(x) \d x$ to be finite, which is not true if $k$ becomes large. However, this issue can be fixed by integrating the TF equation in $\{|x|\le R\}$ and then sending $R\to \infty$ at the end. 
\end{proof}

It is unclear whether the above proof can be used to derive the asymptotic neutrality in Theorem \ref{thm:asympneu}.   

To end this section, let us mention the following analogue of  Conjecture \ref{con:Hartree-hard-time} for  the Thomas--Fermi theory.

\begin{conjecture} \label{con:TF-time} Consider the time-dependent Thomas--Fermi theory 
$$
\begin{cases}
\partial_t \varphi &= \frac{1}{2} (\nabla \varphi)^2 + c_0 \rho^{2/3} - Z|x|^{-1} + \rho* |x|^{-1},\\
\partial_t \rho &= \nabla (\rho \nabla \varphi)
\end{cases}
$$
with a fixed constant $c_0>0$. Then for all initial state $(\varphi_0,\rho_0)$ satisfying 
$0\le \rho_0\in L^1(\R^3)\cap L^{5/3}(\R^3)$ and $\sqrt{\rho_0} |\nabla \varphi_0| \in L^2(\R^3)$, we have 
$$
\limsup_{T\to \infty}\frac{1}{T} \int_0^T \int_{|x|\le R} \rho (x,t) \d x \d t \le Z, \quad \forall R>0. 
$$
\end{conjecture}

In 2018, Chen and Siedentop \cite{CS18} proved the weaker bound $4Z$ instead of $Z$, using a strategy similar to that of \cite{LL13}. Their analysis also covers the  Vlasov equation.

\section{The Thomas--Fermi-von Weizs\"acker  (TFW) theory}

In principle, the TF theory is purely semiclassical and it is only good to describe the bulk of electrons at distance $O(Z^{-1/3})$ to the nucleus. For physical and chemical applications, it is important to capture additional contributions of the {\em innermost} and {\em outermost} electrons, the ones at distances $O(Z^{-1})$ and $O(1)$ to the nucleus, respectively. We refer to Lieb's review \cite{L81} for a pedagogical introduction to several refined versions of the TF theory. 

In this section, we focus on the first refinement: the  TFW  theory  where the ground state energy is given by 
\begin{align}\label{eq:TFW-theory}
E^{\rm TFW}(N)= \inf_{\|u\|_{L^2}^2=N}   \int_{\mathbb{R}^3} \Big( C^{\rm TF} |u|^{10/3} + C^{\rm W} |\nabla u|^2 - \frac{Z|u|^2}{|x|} + \frac{1}{2} |u|^2 \Big(|u|^2*\frac{1}{|x|} \Big) \Big) \d x. 
\end{align}
Here $|u|^2$ plays the role of the electron density and the von Weizs\"acker correction term $C^{\rm W} |\nabla u|^2$, with a constant $C^{\rm W}>0$, corresponds to the contribution of the innermost electrons.  In the context of the ionization conjecture, we have the following theorem. 

\begin{theorem}[\cite{B79,BBL81,BL85}] The variational problem $E^{\rm TFW}(N)$ has a minimizer if and only if $N\le N_c(Z)$ for some critical value $Z<N_c(Z)\le Z+C$.
\end{theorem}

The general framework of the existence, uniqueness and  properties of TFW minimizers was discussed in great detail by Benguria, Brezis and Lieb in 1981 \cite{BBL81}. Although the functional on the right hand side of \eqref{eq:TFW-theory} is still convex in $u$, the TFW theory is significantly more complicated than the TF theory. The analysis in \cite{BBL81} contains several steps; with the starting point being the study of a ``fully unconstrained problem" (namely a version of \eqref{eq:TFW-theory} without any mass constraint of $u$), which is of the same spirit of the above analysis of the Thomas--Fermi theory. The unique fully unconstrained minimizer is a positive, radial solution to the TFW equation
$$
\left( \frac{5}{3}C^{\rm TF} u^{2/3} + C^{\rm W} (-\Delta) - \Phi(x)\right) u (x)=0, \quad \Phi(x)=\frac{Z}{|x|} -  |u|^2*\frac{1}{|x|}
$$ 
and moreover $\int_{\R^3} |u|^2=N_c(Z)$. 

The strict lower bound $N_c(Z)>Z$ shows that the  von Weizs\"acker correction really improves the binding ability of atoms. This remarkable result was proved by Benguria in his 1979 PhD thesis under Lieb's supervision \cite{B79}. The nonexistence part, namely $N_c(Z)\le Z+C$ was proved by Benguria and Lieb in 1984 \cite{BL85}. Let us quickly explain these proofs below. 

\begin{proof}[Proof of $N_c(Z)>Z$ in TFW theory \cite{B79}] Assume that $N_c(Z)=\int_{\R^3} |u|^2 \le Z$. By Newton's theorem, the TFW potential 
$$
\Phi(x)=\frac{Z}{|x|} -  |u|^2*\frac{1}{|x|} = \frac{Z}{|x|} - \int_{\R^3} \frac{|u(y)|^2}{\max(|x|,|y|)} \d y 
$$
is nonnegative. Therefore, from the TFW equation we have
$$
-\Delta u(x) + c_1 u^{5/3}(x) \ge 0 \quad \text{ for all }x\ne 0
$$ 
with a constant $c_1\ge 0$. Consider the function $\widetilde{u}(x)= c_2 |x|^{-3/2}$ with $c_2>0$ sufficiently small such that
\begin{align*}
-\Delta \widetilde{u}(x) + c_1 \widetilde{u} ^{5/3}(x) &\le  0, \quad \forall |x|\ge 1,\\
\widetilde{u}(x) &\le u(x), \quad \forall |x|=1. 
\end{align*}
If the open set $A=\{|x|>1, \widetilde{u}(x) > u(x) \}$ is non empty, then  $\widetilde{u}-u$ is subharmonic and positive in $A$, but vanishes on the boundary of $A$, which is a contradiction to the maximum principle. Thus $u(x)\ge \widetilde u(x)$ for all $|x|\ge 1$, but this contradicts to the fact that $u\in L^2(\R^3)$. Thus $N_c(Z)>Z$.
\end{proof}

\begin{proof}[Proof of $N_c\le Z+C$ in TFW theory \cite{BL85}] The main idea is that the function 
$$ p(x)= (4\pi C^{\rm W} u^2(x) + \Phi^2(x))^{1/2}$$
is subharmonic for $|x|>0$ and $p(x)\to 0$ as $|x|\to \infty$. This implies that $|x|p(x)$ is convex and decreasing in $|x|$. On the other hand, when $|x|\to \infty$, the TFW minimizer $u$ decays faster than any polynomial while the TFW potential satisfies $|x|\Phi(x)\to  -Q(Z)<0$ where $Q(Z)=N_c(Z)-Z$. Therefore, we conclude that 
$$
Q(Z)= \lim_{|y|\to \infty} |y|p(y) \le |x| p(x), \quad \forall |x|>0. 
$$
We can choose $|x_0|\sim O(1)$ such that $\Phi(x_0)<0$ (this follows from $Q(Z)>0$) and $u(x_0)\le C$ (this follows from the TFW equation). Thus $Q(Z) \le C$ as desired. 
\end{proof}

Further important results on the TFW theory were established later by Solovej in his 1989 PhD thesis under Lieb's supervision. In particular, Solovej introduced the  {\em universality} concept, namely some relevant quantities not only are bounded uniformly but also have limits when $Z\to \infty$.  In particular, he proved the following theorem. 

\begin{theorem}[\cite{S90}] \label{thm:S90}The TFW unconstrained minimizer $u_Z$ and the TFW potential $\Phi_Z(x)= Z|x|^{-1}- |u_Z|^2*|x|^{-1}$ have limits when $Z\to \infty$
$$
\lim_{Z\to \infty}u_Z(x) = u_\infty(x) , \quad \lim_{Z\to \infty} \Phi_Z(x)=\Phi_\infty (x),  \quad \forall x\ne 0.
$$
Consequently, the maximum ionization $Q_Z= N_c(Z)-Z$ also has a limit 
$$
\lim_{Z\to \infty} Q_Z = Q_\infty = - \lim_{|x|\to \infty} |x| \Phi_\infty(x). 
$$
\end{theorem}

Recall that the TF theory describes the bulk of the electrons at distance $O(Z^{-1/3})$ to the nucleus, and hence the rescaled function $Z^{-1}u_Z(Z^{-1/3}x)$ has a limit when $Z\to \infty$ which is given by the TF minimizer. However, the universality in Theorem \ref{thm:S90} is much deeper since it describes the outermost electrons at distance $O(1)$ to the nucleus which are responsible for chemical binding. In the level of the many-body Schr\"odinger theory, the convergence of the one-body density $Z^{-2} \rho_Z(Z^{-1/3}x)$ was already proved by Lieb and Simon \cite{LS77}, but the universality remains open. 

\begin{conjecture}[Universality] \label{conj-uni} In the many-body Schr\"odinger theory \eqref{eq:EN}, the one-body density $\rho_Z$ of the ground state with $N=N_c(Z)$ has a limit up to subsequences $Z=Z_n\to \infty$, namely 
$$
\lim_{Z_n\to \infty} \rho_{Z_n}(x) = \rho_\infty(x), \quad \forall x\ne 0.
$$
\end{conjecture} 

Here different subsequences $Z_n\to \infty$ may lead to different limits, which corresponds to the existence of different groups in an ``infinite periodic table". In principle, if Conjecture \ref{conj-uni} holds true, then we should be able to extract the convergence of the maximum ionization as well as the radius of atoms. Currently, the boundedness of these quantities is unknown \cite{S00,LS09}. We refer to a recent paper of Solovej \cite{S17} for further discussions on the universality of large atoms and molecules. 
 
\section{The Hartree--Fock (HF) theory}

In computational physics and chemistry, not only the electron density, but also the electron orbitals are of the fundamental interest. %For this purpose, instead of concentrating on density functional theories like TF or TFW, we need more complicated states to characterize the many-body feature of the system. 
One of the most popular methods in this direction is the Hartree--Fock (HF) theory in which the many-body wave functions are restricted to the Slater determinants
\begin{align}\label{eq:Slater}
\Psi(x_1,...,x_N)= (u_1\wedge ... \wedge  u_N)(x_1,...,x_N)= \frac{1}{\sqrt{N!}} \det [ u_i(x_j)]_{1\le i,j\le N}
\end{align}
where $\{u_i\}_{i=1}^N$ is an orthonormal family in $L^2(\R^3)$. In principle, the Slater determinants are very similar to the Hartree states in \eqref{eq:Hartree}, except that the anti-symmetric tensor has to be taken in \eqref{eq:Slater} to ensure Pauli's exclusion principle \eqref{eq:Pauli}. The Hartree--Fock energy is defined by 
\begin{align}\label{eq:HF}
E^{\rm HF}(N)= \inf_{\Psi \in SD_N} \langle \Psi, H_N \Psi\rangle
\end{align}
where $SD_N$ is the set of $N$-body Slater determinants. Here $N\in \mathbb{N}$ and $Z>0$ is not necessarily an integer.

The analysis of the HF theory is an important subject of mathematical physics. In the context of the ionization problem, the existence of HF miminizers for $N\le Z$ is much harder than that in the many-body Schr\"odinger theory since the set of states is very nonlinear due to the orthogonality of the orbitals $\{u_i\}_{i=1}^N$. This issue was settled in a seminal paper of Lieb and Simon in 1977. 

%The Hartree-Fock theory for Coulomb systems
%Elliott H. Lieb & Barry Simon 
%Communications in Mathematical Physics volume 53, pages185–194 (1977)

\begin{theorem}[\cite{LS77b}] \label{thm:LS77b} For every $N<Z+1$, the HF minimization problem \eqref{eq:HF} has a minimizer. Moreover, the minimizing orbitals $\{u_i\}_{i=1}^N$ are the $N$ lowest eigenfunctions of the one-body operator 
$$
h = -\Delta - Z|x|^{-1}  + U_\Psi(x) - K_\Psi
$$ 
with the multiplication operator $U_\Psi(x)= \sum_{i=1}^N |u_i|^2*|x|^{-1}$ and the Hilbert--Schmidt operator $K_\Psi$ with kernel $K_\Psi(x,y)=\sum_{i=1}^N u_i(x) \overline{u_i(y)} |x-y|^{-1}$. In fact, $h$ has infinitely many negative eigenvalues; in particular $h u_i =\eps_i u_i$ with $\eps_i< 0$ for all $i$. 
\end{theorem}

\begin{proof} The general strategy is to use  the relaxation method in the same spirit of the TF theory. In the HF case, the corresponding ``unconstrained problem" is 
\begin{align}\label{eq:HF-ext}
\inf_{\Psi \in \widetilde{SD}_N} \langle \Psi, H_N \Psi\rangle
\end{align}
where $\widetilde{SD}_N$ contains all $\Psi=u_1\wedge ... \wedge  u_N$ such that the $N\times N$ matrix  
$$M=(\langle u_i, u_j\rangle)_{1\le i,j\le N}$$
 satisfies $0\le M\le 1$. Note that $M=1$ if $\Psi$ is a Slater determinant. The extension to $0\le M\le 1$ makes the set $\widetilde{SD}_N$ stable under the weak convergence in $L^2(\R^3)$ of the orbitals $\{u_i\}_{i=1}^N$, thus ensuring the existence of minimizers of \eqref{eq:HF-ext} by the direct method in the calculus of variations. 

In order to go back to the original problem \eqref{eq:HF}, three key ingredients are needed. First, since the energy $\langle \Psi, H_N \Psi\rangle$ with $\Psi=u_1\wedge ... \wedge  u_N$ is invariant under changing $\{u_i\}_{i=1}^N$ to $\{A u_i\}_{i=1}^N$ with any $N\times N$ unitary matrix $A$, we can assume that $M$ is diagonal, namely 
$$\langle u_i, u_j\rangle = \lambda_{i} \delta_{ij}\quad \text { with } 0\le \lambda_i \le 1.$$ 
Second, note that for each $i$, the function $u=u_i$ is the minimizer of the functional 
$$\Psi= u_1\wedge ... \wedge u_{i-1} \wedge u \wedge u_{i+1}\wedge... \wedge u_N \mapsto \langle \Psi, H_N \Psi\rangle={\rm const}+ \langle u, h u\rangle$$
subject to the constraints
$$\langle u, u_j\rangle=0 \text{ for all }j\ne i, \quad \|u\|_{L^2}\le 1.$$
Therefore, $u_i$ must be a linear combination of the $N$ smallest eigenfunctions of $h$. Up to a further unitary transformation, we can assume that all $u_i$ are eigenfunctions of $h$. Third, when $N<Z+1$, $h$ has infinitely many eigenvalues below its essential spectrum $[0,\infty)$. This  fact can be proved by the min-max principle, using radial trial states with disjoint supports.  Thus for all $i$, we have $hu_i=\eps_i u_i$ with $\eps_i<0$, and hence the minimizing $u_i\mapsto \langle u_i, h_i u_i\rangle = \eps_i \|u_i\|_{L^2}^2$ under the constraint $\|u_i\|_{L^2}\le 1$ must satisfy $\|u_i\|_{L^2}=1$. 
\end{proof}

Note that a Slater determinant can be encoded fully in terms of its {\em one-body density matrix}. Recall that for every $N$-body wave function $\Psi$, the one-body density matrix $\gamma_\Psi$ is a trace class operator on $L^2(\R^3)$ with kernel  
\begin{align} \label{eq:1pdm}
\gamma_\Psi(x,y)=N \int_{\R^{3(N-1)}}\int_{\R^3} \Psi(x,x_2,...,x_N)\overline{\Psi(y,x_2,...,x_N)} \d x_2 ... \d x_N.
\end{align}
In particular, if $\Psi$ is given in \eqref{eq:Slater}, then $\gamma_\Psi$ is the rank-$N$ orthogonal projection 
$$
\gamma_\Psi = \sum_{i=1}^N |u_i\rangle \langle u_i|.  
$$
The one-body density $\rho_\Psi$ defined in \eqref{eq:rhoN} is given equivalently by $\rho_\Psi(x)=\rho_\gamma(x)=\gamma(x,x)$. Using these notations, the energy of a Slater determinant $\Psi$ is given by  
$$
\langle \Psi, H_N \Psi\rangle = \mathcal{E}^{\rm HF}(\gamma_\Psi) 
$$
where 
$$
\mathcal{E}^{\rm HF}(\gamma)= {\rm Tr} ((-\Delta - Z|x|^{-1} )\gamma) +  \frac{1}{2} \iint  \frac{\rho_\gamma(x)\rho_\gamma(y) - |\gamma(x,y)|^2}{|x-y|} \d x \d y.
 $$
 Consequently, the   Hartree--Fock energy in \eqref{eq:HF} can be written equivalently as  
\begin{align}\label{eq:HF-gamma}
E^{\rm HF}(N)=\inf_{\substack {0\le \gamma=\gamma^2 \le 1\\ {\rm Tr} \gamma =N} } \mathcal{E}^{\rm HF}(\gamma) .%$$
\end{align}
In this direction, the relaxation method suggests to relate \eqref{eq:HF-gamma} to the ``unconstrained problem" 
\begin{align}\label{eq:HF-gamma-0-1}
E^{\rm HF}_{\le}(N)=\inf_{\substack {0\le \gamma \le 1\\ {\rm Tr} \gamma =N} } \mathcal{E}^{\rm HF}(\gamma).  
%$$
\end{align}
Here we drop the projection  condition $\gamma=\gamma^2$ in \eqref{eq:HF-gamma-0-1} in order to make the set of states convex. Thus in principle, the unconstrained energy $E^{\rm HF}_{\le}(N)$ is much easier to compute than the original energy $E^{\rm HF}(N)$. On the other hand, while $E^{\rm HF}(N)$ is an obvious upper bound to the full many-body $E_N$ in \eqref{eq:EN}, it is unclear if the unconstrained energy $E^{\rm HF}_{\le}(N)$ has this nice property or not. This conceptual difficulty was removed completely in 1981 by Lieb. 

\begin{theorem}[Lieb's variational principle \cite{L81b}] \label{thm:L81b} Let $0\le \gamma \le 1$ and ${\rm Tr} \gamma=N$. Then there exist an $N$-body Slater determinant $\Psi$  and an $N$-body mixed state $\Gamma$ such that its one-body density matrix is $\Gamma^{(1)}=\gamma$ and  
$$
\langle \Psi, H_N \Psi\rangle \le {\rm Tr}(H_N \Gamma) \le \mathcal{E}^{\rm HF}(\gamma).
$$
\end{theorem}

Here $\Gamma$ is an $N$-body mixed state if $\Gamma= \sum_{i} \lambda_i |\Psi_i\rangle \langle \Psi_i|$ with $N$-body orthonormal functions $\{\Psi_i\}$ and nonnegative numbers $\{\lambda_i\}$ satisfying $\sum_i \lambda_i=1$. In terms of the one-body density matrices,  we have $\Gamma^{(1)}=\sum_i \lambda_i \gamma_{\Psi}$ where $\gamma_\Psi$ is defined in 
\eqref{eq:1pdm}.

A direct consequence of Lieb's theorem is that $E^{\rm HF}_{\le}(N)=E^{\rm HF}(N)$, which makes the formulation \eqref{eq:HF-gamma-0-1} extremely helpful to compute an energy upper bound of $E_N$ by the trial state argument. 

As mentioned in \cite{L81b}, Theorem \ref{thm:L81b} holds for any two-body interaction which is positive semidefinite. We refer to Bach's paper \cite{B92} for a simplified proof of this result. Theorem \ref{thm:L81b} is one of the main tools in Bach's proof that the HF energy agrees with the best known expansion of the quantum energy, namely  
$$
E_N = c_1 Z^{7/3} + c_2 Z^{2} + c_3 Z^{5/3} + o(Z^{5/3}) = E_N^{\rm HF}+o(Z^{5/3}), \quad \forall N\in [Z,N_c(Z)]
$$
where the first equality was established previously by Fefferman and Seco \cite{FS90b}.

Using the concept of one-body density matrices, we can rewrite the proof of  Theorem \ref{thm:LS77b} as follows. 

\begin{proof}[A shorter proof of Theorem \ref{thm:LS77b}] Consider the variational problem
\begin{equation}\label{eq:EHF-rel-rel} \widetilde{E}^{\rm HF}_{\le}(N) = \inf_{\substack {0\le \gamma \le 1\\ {\rm Tr} \gamma \le N} } \mathcal{E}^{\rm HF}(\gamma).
\end{equation}
The existence of a minimizer $\gamma$ of \eqref{eq:EHF-rel-rel} can be proved by  the direct method in the calculus of variations (the set of states is stable under the weak-*  convergence in trace class). For every $0\le \widetilde{\gamma}\le 1$ with ${\rm Tr}\widetilde{\gamma} \le N$, the function $t\mapsto \mathcal{E}^{\rm HF}( (1-t)  \gamma + t\widetilde \gamma)$ with $t\in [0,1]$ attains its minimum at $t=0$,  and hence
$$
0 \le \frac{d}{dt} \mathcal{E}^{\rm HF}( (1-t)  \gamma + t\widetilde \gamma)|_{t=0}= {\rm Tr} (h (\widetilde{\gamma}-\gamma))
$$
where $h=-\Delta - Z|x|^{-1}+\rho_\gamma*|x|^{-1}-K_\gamma$ with $K_\gamma(x,y)=\gamma(x,y)/|x-y|$. 

If $N<Z+1$, then $h$ has infinitely many negative eigenvalues $\eps_1\le \eps_2\le ...<0$ (this was already explained in the previous proof).  Consequently, we can choose $\widetilde \gamma$ to be the projection on the lowest $N$ eigenfunctions of $h$, so that
$${\rm Tr}(h\gamma)\le {\rm Tr} (h \widetilde \gamma)=\sum_{i=1}^N \eps_i.$$
Since $0\le \gamma \le 1$ and ${\rm Tr}\gamma\le N$, this implies that ${\rm Tr}\gamma=N$. Thus $\gamma$ is also a minimizer for $E^{\rm HF}_{\le}(N)$, and by Theorem \ref{thm:L81b} the existence of minimizers of $E^{\rm HF}(N)$ follows. 
\end{proof}

%By the variational principle, the HF energy $E^{\rm HF}(N)$ is clearly an upper bound to the quantum energy $E_N$. In 1992, Bach \cite{B92} proved a rigorous lower bound showing that the HF energy agrees with the best known expansion of the quantum energy in  \cite{FS90}.  

Now let us turn to the nonexistence in the HF theory. All non-asymptotic bounds in Section \ref{sec:2} extend to the HF case without significant modifications; in particular Lieb's Theorem \ref{thm:L84} ensures that $N_c(Z)<2Z+1$. However, the conjecture bound $N_c(Z)\le Z +C$ had been open for a long time until solved by Solovej in 2003.  

\begin{theorem}[\cite{S03}] $N_c(Z)\le Z+C$ in the Hartree--Fock theory. 
%There exists a universal constant $C>0$  such that if $N>Z+C$, then  $E^{\rm HF}(N)$ has no minimizer. 
\end{theorem} 

The proof in \cite{S03} is based on a clever use of the Benguria--Lieb method, but only for outermost electrons. More precisely, assuming that we have an efficient method to separate $m$ outermost electrons from the rest of the system, which is of the effective charge $Z'= Z - (N-m)$, then the  Benguria--Lieb method gives $m < 2Z' + 1$. Since $Z'$ is smaller than $Z$, the loss of the factor 2 becomes less serious. Solovej's idea is to propose a rigorous bootstrap argument to bring $Z'$ down to order 1 after finitely many steps. On the technical level, the key tool in \cite{S03} is a rigorous comparison between the HF potential 
$$
\Phi^{\rm HF}_Z(x) = \frac{Z}{|x|} - \int_{|y|\le |x|} \frac{\rho^{\rm HF}(y)}{|x-y|} dy.  
$$ 
and the corresponding TF potential $\Phi^{\rm TF}_Z(x)$, namely 
\begin{align} \label{eq:HF-TF}
|\Phi_Z^{\rm HF}(x)  - \Phi^{\rm TF}_Z(x)| \le C (1+ |x|^{-4+\varepsilon}), \quad \forall x\ne 0
\end{align}
for some universal constants $C>0$, $\varepsilon>0$. Note that $\Phi^{\rm TF}_Z(x)$ behaves as $|x|^{-4}$ for $|x|\gg Z^{-1/3}$. The significance of \eqref{eq:HF-TF} is that the TF theory captures correctly the HF theory, at least in terms of the potentials, up to a length scale of order 1. This is highly remarkable since due to its semiclassical nature the TF theory is supposed to be good only for $|x|\sim Z^{-1/3}$. This property suggests that the universality in Conjecture \ref{conj-uni} should hold also in the HF theory, but a rigorous proof is still missing.

In Solovej's strategy,  the main conceptual difficulty is the splitting of ``problem from outside" from the ``problem from inside". In the HF theory, this can be done  using the unconstrained formulation \eqref{eq:HF-gamma-0-1} and Lieb's Theorem \ref{thm:L81b}. Unfortunately, this technique is not available on the level of the many-body Schr\"odinger theory. It seems that a completely new many-body localization technique which be needed to solve the ionization conjecture.

\section{Liquid drop model}

In 1928, Gamow proposed a theory to describe a nucleus using only the number of nucleons (protons and neutrons) and the electrostatic energy of protons. This problem has gained renewed interest from many mathematicians \cite{CMT17}. To be precise, the liquid drop model is associated to the minimization problem 
\begin{align} \label{eq:EG}
E^{\rm G}(m)= \inf_{|\Omega|=m} \mathcal{E}(\Omega) 
\end{align}
where
$$
\mathcal{E}(\Omega)  =  {\rm Per}(\Omega) + D(\Omega) = {\rm Per}(\Omega) + \frac{1}{2}\int_{\Omega} \int_{\Omega}\frac{1}{|x-y|} dx dy.
$$
Here $\Omega\subset \R^3$ stands for the nucleus and ${\rm Per}(\Omega)$ is the perimeter in the sense of De Giorgi (which is the surface area of $\Omega$ when the boundary is smooth). 

It is generally assumed in the physics literature that if a minimizer exists, then it is a ball. Consequently, by comparing the energy of a ball of volume $m$ with the energy of a union of two balls of volume $m/2$, one expects the nonexistence of minimizers if $m>m_*$ with
$$m_*=5\frac{2-2^{2/3}}{2^{2/3}-1}\approx 3.518.$$

\begin{conjecture}[\cite{CP11}] \label{conj:li}$E^{\rm G}(m)$ has a minimizer if and only if $m\le m_*$.  Moreover, if a minimizer exists, then it is a ball. 
%when $m\le $
\end{conjecture}

The question here is somewhat similar to that of the ionization conjecture of atoms. As we will see, some ideas from the liquid drop model turn out to be helpful for the ionization problem.

 On the mathematical side, among all measurable sets of a given volume, although a ball minimizes the perimeter (by the isoperimetric inequality \cite{DG58}), it does maximize the Coulomb self-interaction energy (by the Riesz rearrangement inequality \cite{R30}). Therefore, it is unclear why balls should be the minimizers. Consequently, the argument predicting the threshold $m_*$ is questionable. 
 
 In 2014,  Kn\"upfer and Muratov \cite{KM14} proved that if $m>0$ is sufficiently small, then $E^{\rm G}(m)$ has a unique minimizer which is a ball. The proof in \cite{KM14} uses deep techniques in geometric measure theory, including a quantitative
isoperimetric inequality of Fusco, Maggi, and Pratelli \cite{FMP08}. 

On the other hand, it is desirable to develop a non-perturbative approach to handle larger masses. In 2015, Frank and Lieb proved the following result, which serves as a basic tool to analyze the existence question for all $m>0$. 

\begin{theorem}[\cite{FL15}] \label{thm:FL15} If for a given $m>0$, one has the strict binding inequality
\begin{align}\label{eq:FL-strict-ineq}
E^{\rm G}(m)< E^{\rm G}(m-m') + E^{\rm G}(m'), \quad \forall 0<m'<m,
\end{align}
then $E^{\rm G}(m)$ has a minimizer. 
\end{theorem} 

Theorem \ref{thm:FL15} can be interpreted in the same spirit of the strict binding inequality $E(N)<E(N-1)$ in the context of the ionization problem. The proof in  \cite{FL15} is based on the ``method of the missing mass", which goes back to Lieb's 1983 work on sharp Hardy--Littlewood--Sobolev and related inequalities \cite{L83}. 

Very recently, Frank and Nam \cite{FN21}  used Theorem \ref{thm:FL15} to establish the {\em optimal existence} in Conjecture \ref{conj:li}. 

\begin{theorem}[\cite{FN21}] \label{thm:FL21} $E^{\rm G}(m)$ has a minimizer for every $0<m\le m_*$.
\end{theorem} 

\begin{proof} Let us prove the strict binding inequality \eqref{eq:FL-strict-ineq} for all $0<m<m_*$. Let $0<m_1<m$ and  $s= m_1/m \in (0,1)$. As a first step, we observe that if $|\Omega|=m_1$, then $|s^{-1/3} \Omega|=m$, and hence by the variational principle
\begin{align*}
E^{\rm G}(m) \le \mathcal E (s^{-1/3} \Omega) = s^{-2/3}  {\rm Per} \Omega + s^{-5/3} D (\Omega) = s^{-5/3}  \mathcal E(\Omega) - s^{-5/3}( 1  - s) {\rm Per} \ \Omega.
\end{align*}
On the other hand, by the isoperimetric inequality
$$
{\rm Per} \ \Omega \ge m_1^{2/3} {\rm Per} B_1 = s^{2/3}  m^{2/3} {\rm Per} B_1
$$
where $B_1$ is the ball of volume 1 in $\R^3$ . Inserting this in the above inequality, optimizing over $\Omega$, and rearranging terms we find that
$$
E(m_1) \ge s^{5/3} E(m) +   s^{2/3} (1-s)  m^{2/3} {\rm Per} B_1.
$$
Similarly,
$$
E(m-m_1) \ge (1-s)^{5/3} E(m) +  (1-s)^{2/3}  s  m^{2/3} {\rm Per} B_1.
$$
Therefore,
 \begin{align*} 
& E(m_1)+  E(m-m_1) - E(m)  \\
& \ge (  s^{5/3} +  (1-s)^{5/3} - 1) E(m) + \Big( s^{2/3} (1-s) + (1-s)^{2/3}  s \Big) m^{2/3} {\rm Per} \ B_1.
 \end{align*}
Using $s^{5/3} +  (1-s)^{5/3} - 1<0$ and $E(m)\le  \mathcal E(m^{1/3}B_1)$ we find that 
  \begin{align*} %\label{eq:ldm-binding-0}
& E(m_1)+  E(m-m_1)- E(m)  \nonumber \\
&\ge  \Big(  s^{5/3} + (1-s)^{5/3} - 1  \Big)  \Big( D(B_1)  m  -  f(s) {\rm Per} \ B_1   \Big) m^{2/3}  
 \end{align*}
 with 
\begin{equation*} % \label{eq:def-f}
 f(s):= \frac{s^{2/3} +  (1-s)^{2/3} -1  }{1-s^{5/3} - (1-s)^{5/3} }. 
\end{equation*}
Therefore, $E(m_1)+  E(m-m_1)- E(m)>0$  if  
\begin{equation} \label{eq:gs>m}
m< \frac{{\rm Per} \ B_1}{D(B_1)} \min_{s\in [0,1]} f(s).
\end{equation}
A direct computation shows that the right hand side of \eqref{eq:gs>m} coincides with $m_*$. 

Thus the existence of minimizers for every $0<m < m_*$ follows immediately from Theorem \ref{thm:FL15}. The existence can be extended to $m=m^*$ by a continuity argument from \cite[Theorem 3.4]{FL15}. 
 \end{proof}

In \cite{FN21}, we also proved that if the nonexistence in Conjecture \ref{conj:li} holds,  then the above proof can be refined to show that the minimizer for $m<m_*$ is unique and it is a ball. Thus only the (optimal) nonexistence part is missing.  

There are some partial nonexistence results. Kn\"upfer and Muratov \cite{KM14} proved that $E^{\rm G}(m)$ has no minimizer if  $m$ is sufficiently large. The same result was proved by Lu and Otto \cite{LO14} by a different method.  This result is comparable to the bound $N_c(Z)<\infty$ in the ionization problem. In 2016, Frank,  Killip and Nam \cite{FKN16} proved the nonexistence for all $m>8$, which is somewhat comparable to Lieb's bound $N_c(Z)<2Z+1$ in the ionization problem. 

 \begin{proof}[Proof of nonexistence for $m>8$ \cite{FKN16}.] Assume that $E^{\rm G}(m)$ has a minimizer $\Omega$. We split $\Omega$ into two parts, $
\Omega= \Omega^+ \cup \Omega^- 
$, by  a hyperplane $H$ and then move $\Omega^-$  to infinity by translations. Since $\Omega$ is a minimizer, we obtain the binding inequality 
 \begin{align*}
{\rm Per}(\Omega) + \int_{\Omega} \int_{\Omega}\frac{1}{|x-y|} dx dy &\le {\rm Per}(\Omega^+) + \int_{\Omega^+} \int_{\Omega^+}\frac{1}{|x-y|} dx dy \\
&\quad +{\rm Per}(\Omega^-) + \int_{\Omega^-} \int_{\Omega^-}\frac{1}{|x-y|} dx dy
\end{align*}
which is equivalent to
$$
2 \mathcal{H}^2 (\Omega\cap H) \ge  \int_{\Omega^+} \int_{\Omega^-}\frac{1}{|x-y|} dx dy.  
$$
Here $\mathcal{H}^2$ is the two-dimensional Hausdorff measure. Next, we parameterize: 
$$
H= H_{\nu,\ell} = \{ x\in \mathbb{R}^3: \, x \cdot \nu =\ell\}
$$
with $\nu\in S^2$, $\ell\in \mathbb{R}$. The above inequality becomes
$$
2 \mathcal{H}^2 (\Omega\cap H_{\nu,\ell}) \ge   \int_{\Omega} \int_{\Omega}\frac{\chi(\nu\cdot x > \ell > \nu \cdot y)}{|x-y|} dx dy.
$$
Integrating over $\ell\in \mathbb{R}$ and using Fubini's theorem we get
$$
2 |\Omega|  \ge   \int_{\Omega} \int_{\Omega}\frac{[\nu\cdot (x-y)]_+}{|x-y|} dx dy.
$$
Finally, averaging over $\nu \in S^2$ and using 
$$
\int [\nu \cdot z]_{+} \frac{d\nu}{4\pi} = \frac{|z|}{2} \int_0^{\pi/2} \cos \theta \sin \theta d\theta =\frac{|z|}{4} 
$$
with $z=x-y$, we conclude that $2|\Omega|\ge \frac 1 4 |\Omega|^2$, namely $|\Omega| \le 8$.
\end{proof}

It is interesting that the above cutting argument can be used to replace the Benguria--Lieb argument in the ionization problem in various situations. In 2018, Frank, Nam and Van Den Bosch \cite{FNV18} used this technique to establish the ionization conjecture in the Thomas--Fermi--Dirac-von Weis\"acker theory. In this model, the standard  Benguria--Lieb method does not apply due to Dirac's correction term to the exchange energy, but a modification of the above cutting argument gives an efficient control of number of particles ``outside" in terms of particles ``inside", thus enabling us to employ Solovej's bootstrap argument as in the HF theory. 

In  \cite{FNV18b}, we extended the nonexistence $N_c(Z)\le Z+C$ to the M\"uller density-matrix-functional theory. In this model, the existence for $N\le Z$ was proved by Frank, Lieb, Seiringer and Siedentop in 2007 \cite{FLSS07} using a relaxation method in the spirit of TF and HF theories. In \cite{K17}, Kehle established the nonexistence  to a family of density-matrix-functional theories  that interpolates the HF and M\"uller theories. 

Hopefully an exchange of ideas from the ionization problem to the liquid drop model will lead to further results in the future.


\begin{thebibliography}{19}

\bibitem{B92} V. Bach. Error Bound for the Hartree-Fock Energy of Atoms and Molecules. Commun. Math. Phys. 147 (1992), pp. 527--548. 


\bibitem{B84} B. Baumgartner. On Thomas-Fermi-von Weizs\"acker and Hartree energies as
functions of the degree of ionization. J. Phys. A: Math. Gen.  17(1984), pp. 1593--1602.

\bibitem{B79} R. Benguria. The von Weizs\"acker and exchange corrections in the Thomas Fermi theory. PhD thesis, Princeton University, 1979. 


\bibitem{BBL81} R. Benguria, H. Brezis, and E. H. Lieb, The Thomas-Fermi-von Weizs\"acker theory of atoms and molecules. Commun. Math. Phys. 79 (1981), pp. 167--180. 


\bibitem{BL83} R. Benguria and E. H. Lieb. Proof of stability of highly negative ions in the absence of the Pauli
principle. Phys. Rev. Lett.   50 (1983), pp. 1771--1774.

\bibitem{BL85} R. Benguria and E. H. Lieb. The most negative ion in the Thomas-Fermi-von Weizs\"acker
theory of atoms and molecules. J. Phys. B 18 (1985), pp. 1045--1059.

\bibitem{BT22} R. Benguria and T. Tubino. Analytic bound on the excess charge for the Hartree Model. arXiv:2201.13421


\bibitem{CS13} L. Chen and H. Siedentop. Positivity of $|p|^a |q|^b+ |q|^b |p|^a$. J. Funct. Anal. 264 (2013),  pp. 2817--2824.

\bibitem{CS18} L. Chen and H. Siedentop. The maximal negative ion of the time dependent Thomas-Fermi and the
Vlasov atom.  J. Math. Phys. 59 (2018), p. 051902.



\bibitem{CP11} R. Choksi and M. A. Peletier. Small volume-fraction limit of the diblock copolymer problem: II. Diffuse-interface functional. SIAM J. Math. Anal. 43  (2011),  pp. 739--763.

\bibitem{CMT17} R. Choksi, C. B. Muratov, I. Topaloglu. An Old Problem Resurfaces Nonlocally: Gamow's Liquid Drops Inspire Today's Research and Applications.  Notices of the AMS, December 2017. 


\bibitem{DG58} E. De Giorgi. Sulla propriet\`a isoperimetrica dell'ipersfera, nella classe degli insiemi aventi frontiera orientata di misura finita. Atti Accad. Naz. Lincei. Mem. Cl. Sci. Fis. Mat. Nat. Sez. I (8) (1958), pp. 33--44.


\bibitem{FS90b} C. Fefferman and L. A. Seco. On the energy of a large atom. Bull. Amer. Math. Soc.  23 (1990), pp. 525--530.

\bibitem{FS90} C. Fefferman and L. A. Seco. Asymptotic neutrality of large ions. Commun. Math. Phys. 128 (1990), pp. 109-130.

\bibitem{FMP08} N. Fusco, F. Maggi, and A. Pratelli. The sharp quantitative isoperimetric inequality. Ann. of Math. 168 (2008), pp. 941--980. 

\bibitem{FKN16} R. L. Frank, R. Killip, and P. T. Nam. Nonexistence of large nuclei in the liquid drop model. Lett. Math. Phys. 106 (2016), pp. 1033--1036. 

\bibitem{FL15} R. L. Frank and E. H. Lieb. A Compactness Lemma and Its Application to the Existence of Minimizers for the Liquid Drop Model.  SIAM J. Math. Anal., 47 (2015), no. 6, 4436--4450.

\bibitem{FN21} R. L. Frank and P. T. Nam. Existence and nonexistence in the liquid drop model. Calc. Var. Partial Differential Equations, 60, 223 (2021), pp. 1-12.

\bibitem{FLSS07} R. L. Frank, E. H. Lieb, R. Seiringer, and H. Siedentop. M\"uller’s exchange-correlation energy in density-matrix-functional theory. Phys. Rev. A 76 (2007), p. 052517. 




\bibitem{FNV18} R. L. Frank, P. T. Nam and H. Van Den Bosch. The ionization conjecture in Thomas--Fermi--Dirac--von Weizs\"acker theory. Comm. Pure Appl. Math. 71 (2018), pp. 577--614.

\bibitem{FNV18b} R. L. Frank, P. T. Nam, and H. Van Den Bosch.  The maximal excess charge in M\"uller density-matrix-functional theory.  Ann. Henri Poincar\'e 19  (2018), pp. 2839--2867. 

\bibitem{G77} R. T. Glassey. On the blowing up of solutions to the Cauchy problem for nonlinear Schr\"odinger equations. J. Math. Phys. 18 (1977), pp. 1794--1797.

%\bibitem{G21} Y. Goto. The maximal excess charge in reduced Hartree--Fock molecule. Rev. Math. Phys. 33 (2021), p. 2150008


\bibitem{HO77} M. Hoffmann-Ostenhof and T. Hoffmann-Ostenhof. ``Schr\"odinger inequalities" and asymptotic behavior of the electron density of atoms and molecules. Phys. Rev. A 16 (1977), pp. 1782--1785.

\bibitem{K17} C. Kehle. The maximal excess charge for a family of density-matrix-functional theories including Hartree-Fock and M\"uller theories
J. Math. Phys. 58 (2017), p. 011901.  


\bibitem{KM14} H. Kn\"upfer, C. B. Muratov. On an isoperimetric problem with a competing nonlocal term
II: The general case. Comm. Pure Appl. Math. 67 (2014), pp. 1974 --1994.

\bibitem{LL13} E. Lenzmann and M. Lewin. Dynamical Ionization Bounds for Atoms. Anal. and PDE. 6 (2013), pp. 1183-1211.

\bibitem{LNR14} M. Lewin, P. T. Nam, and N. Rougerie. Derivation of Hartree's theory for generic mean-field Bose systems. Advances in Math. 254 (2014), pp. 570--621. 



\bibitem{L81} E. H. Lieb. Thomas-Fermi and related theories of atoms and molecules. Rev. Mod. Phys. 53 (1981), pp. 603-641.

\bibitem{L83} E. H. Lieb. Sharp Constants in the Hardy-Littlewood-Sobolev and Related Inequalities. Ann. of Math. 118 (1983), pp. 349--374.

\bibitem{L81b} E. H. Lieb. Variational Principle for Many-Fermion Systems. Phys. Rev. Lett. 46 (1981), p. 457. 



\bibitem{L84b} E. H. Lieb. Atomic and Molecular Negative Ions, Phys. Rev. Lett. 52 (1984), pp. 315--317.  


\bibitem{L84} E. H. Lieb. Bound on the maximum negative ionization of atoms and molecules. Phys. Rev. A  29 (1984), pp. 3018--3028.

\bibitem{LO81} E. H. Lieb and S. Oxford. Improved lower bound on the indirect Coulomb energy.  International Journal of Quantum Chemistry. 19 (1981), p. 427. 

\bibitem{LSST84} E. H. Lieb, I. M. Sigal, B. Simon, and W. Thirring. Asymptotic Neutrality of Large-Z Ions. Phys. Rev. Lett. 52 (1984), p. 994996.  

\bibitem{LSST88} E. H. Lieb, I. M. Sigal, B. Simon, and W. Thirring. Asymptotic neutrality of large-Z ions. Commun. Math. Phys. 116 (1988), pp. 635-644.

\bibitem{LS73} E. H. Lieb and B. Simon. Thomas-Fermi Theory Revisited. Phys. Rev. Lett. 31 (1973), p. 681 

\bibitem{LS77} E. H. Lieb and B. Simon. The Thomas-Fermi theory of atoms, molecules and solids. Advances in Math. 23 (1977),  pp. 22--116.

\bibitem{LS77b} E. H. Lieb and B. Simon. The Hartree-Fock theory for Coulomb systems. Commun. Math. Phys. 53 (1977), pp. 185--194.

\bibitem{LS09} E. H. Lieb and R. Seiringer. The stability of matter in quantum mechanics. Cambridge University Press, 2009.

\bibitem{LT75} E. H. Lieb and W. Thirring. Bound for the Kinetic Energy of Fermions which
Proves the Stability of Matter. Phys. Rev. Lett. 35 (1975), pp. 687--689.

\bibitem{LO14} J. Lu and F. Otto. Nonexistence of a Minimizer for Thomas--Fermi--Dirac--von Weizs\"acker Model.  Comm. Pure Appl. Math. 67 (2014), pp. 1605-1617. 

\bibitem{M68} C. S. Morawetz. Time decay for the nonlinear Klein-Gordon equations, Proc. Roy. Soc. Ser. A. 306 (1968), pp. 291--296.


\bibitem{N12} P.~T. Nam. New bounds on the maximum ionization of atoms.   Commun. Math. Phys. 312 (2012), pp. 427--445.
  
\bibitem{N13} P. T. Nam. On the number of electrons that a nucleus can bind. Proceedings of the ICMP 2012, A. Jensen (ed.), World Scientific 2013, pp. 504-511. 

\bibitem{N21} P. T. Nam. The Ionization Problem. EMS Newsletter, December 2020, pp. 22--27.  

\bibitem{R30} {F. Riesz},  Sur une in\'egalit\'e int\'egrale. J. London Math. Soc. 5 (1930), pp. 162--168.


\bibitem{R82} M. B. Ruskai. Absence of discrete spectrum in highly negative ions. Commun. Math. Phys. 82 (1982), pp. 457--469.

\bibitem{R82b} M. B. Ruskai. Absence of discrete spectrum in highly negative ions, II. Extension to Fermions. Commun. Math. Phys. 82 (1982), pp. 325--327.

\bibitem{SSS90} L. A. Seco, I. M. Sigal, and J. P. Solovej. Bound on the ionization energy of large atoms. Commun. Math. Phys. 131 (1990), pp. 307--315.

\bibitem{S82} I. M. Sigal. Geometric methods in the quantum many-body problem. Nonexistence of very negative ions. Commun. Math. Phys. 85 (1982), pp. 309--324.

\bibitem{S84} I.M. Sigal, How many electrons can a nucleus bind?  Ann. Phys. 157 (1984), pp. 307--320.

\bibitem{S00} B. Simon. Schr\"odinger operators in the twenty-first century. Mathematical Physics 2000, Imperial College Press, pp. 283-288.


\bibitem{S90} J.P. Solovej. Universality in the Thomas-Fermi-von Weizs\"acker model of atoms and molecules. Commun. Math. Phys.129 (1990),  pp. 561--598.  

\bibitem{S90b} J.P. Solovej. Asymptotics for bosonic atoms. Lett. Math. Phys.  20 (1990), pp. 165--172.

\bibitem{S03} J. P. Solovej. The ionization conjecture in Hartree--Fock theory. Ann. of Math. 158 (2003), pp. 509--576.

\bibitem{S17} J. P. Solovej. A new look at Thomas--Fermi theory.  Molecular Physics 114 (2016), pp. 1036--1040.

\bibitem{T09} G. Teschl. Mathematical methods in quantum mechanics, with applications to Schr\"odinger operators. Graduate Studies in Mathematics, Vol. 99. Providence, RI, Amer. Math. Soc., 2009.

\bibitem{T08} T. Tao. A global compact attractor for high-dimensional defocusing non-linear Schr\"odinger equations with potential, Dyn. Partial Differ. Equ. 5 (2008), pp. 101--116.



\bibitem{Z60} G. Zhislin, Discussion of the spectrum of Schr\"odinger operator for system of many particles, {Trudy. Mosk. Mat. Ob\v s\v c.} {9}, 81 (1960). 


\end{thebibliography}
\end{document}